\documentclass[11pt]{article}

\usepackage{amsmath,amssymb,amsthm,graphicx,color}

\newtheorem{theorem}{Theorem}

\newtheorem{remark}[theorem]{Remark}
\newtheorem{lemma}[theorem]{Lemma}
\newtheorem{proposition}[theorem]{Proposition}

\numberwithin{equation}{section}
\numberwithin{theorem}{section}
\title{On the probability of positive-definiteness in the gGUE via semi-classical Laguerre polynomials}
\author{Alfredo Dea\~{n}o\footnotemark[1]\, and Nicholas J. Simm\footnotemark[2]}
\date{\today}

\begin{document}
\maketitle
\renewcommand{\thefootnote}{\fnsymbol{footnote}}
\footnotetext[1]{School of Mathematics, Statistics and Actuarial Science.
University of Kent. Canterbury CT2 7FS, UK.
A.Deano-Cabrera@kent.ac.uk}
\footnotetext[2]{Mathematics Institute.
University of Warwick. Coventry CV4 7AL, UK. 
n.simm@warwick.ac.uk}

\begin{abstract}
In this paper, we compute the probability that an $N \times N$ matrix from the generalised Gaussian Unitary Ensemble (gGUE) is positive definite, extending a previous result of Dean and Majumdar \cite{DM}. For this purpose, we work out the  large degree asymptotics of semi-classical Laguerre polynomials and their recurrence coefficients, using the steepest descent analysis of the corresponding Riemann--Hilbert problem.
\end{abstract}

\section{Introduction and main results}

The \textit{Gaussian Unitary Ensemble} (GUE) is the most classical and studied example of a unitarily invariant Hermitian random matrix ensemble. Given the set of $N \times N$ Hermitian matrices $\mathcal{H}_{N}$, one defines a probability measure
\begin{equation}
dP(M_N)=\frac{1}{\mathcal{Z}_N} e^{-N\, \mathrm{Tr} V(M_N)}dM_N, \label{gendens}
\end{equation}
where $dM_{N}$ is the Lebesgue measure on the independent entries of $M_{N}$, and $\mathcal{Z}_N$ is the partition function:
\begin{equation}\label{general_ZN}
\mathcal{Z}_N=\int_{\mathcal{H}_{N}} e^{-N\, \mathrm{Tr} V(M_N)}dM_N.
\end{equation}

The potential $V$ is a smooth function with sufficient growth at infinity, so that \eqref{general_ZN} is well defined, and the GUE corresponds to the quadratic case $V(x)=x^{2}$, see references \cite{AGZ, Forrester, Mehta} for relevant background.

In this paper we are interested in the probability that matrices drawn at random from \eqref{gendens} are positive definite, denoted here by $\mathbb{P}(M_{N}>0)$. As well as being a natural question within random matrix theory, in several situations in the physics literature $M_{N}$ is used to model the Hessian matrix of random high-dimensional \textit{energy surfaces}, see \textit{e.g.} \cite{AAE, CGG, DM, Fyodorov} and references therein. In such contexts $\mathbb{P}(M_{N}>0)$ provides important information on the stability (maxima and minima) of such energy surfaces. Outside physics, this question turns out to appear explicitly in certain number theoretical problems \cite{BCFJK}.

In the GUE case, the earliest investigations of this probability appeared in the string theory and cosmology literature \cite{AAE}, where it was argued that $\mathbb{P}(M_N>0)$ decays exponentially in $N^{2}$ (at least implicitly, this already followed from the large deviations principle of Ben Arous and Guionnet \cite{BG}). However, the multiplicative constant in these asymptotics remained unknown until the work of Dean and Majumdar \cite{DM}, who showed using Coulomb gas techniques that
\begin{equation}
\log \mathbb{P}(M_{N}>0) = -c_{1}N^{2}+o(N^{2}) \label{dmprob}
\end{equation}
where
\begin{equation}\label{c1}
c_{1}=\frac{\log 3}{2}.
\end{equation}
Then in subsequent work \cite{BEMN11}, further terms in the asymptotic expansion of \eqref{dmprob} were computed using the technique of loop equations, where it was shown that
\begin{equation}
	\log \mathbb{P}(M_{N}>0) = -c_{1}N^{2}+c_{2}\log(N)+c_{3}+o(1) \label{loopresult}
\end{equation}
where 
\begin{equation}\label{c2c3}
c_{2}=-\frac{1}{12}, \qquad
c_{3} = \frac{\log 3}{8}-\frac{\log 2}{6}+\zeta'(-1),
\end{equation}
and $\zeta(s)$ is the Riemann zeta function. Under certain technical assumptions, Borot and Guionnet \cite{BG13} later developed a rigorous mathematical basis for the application of the loop equations, so that \eqref{loopresult} can be said to constitute a rigorously established result. Our aim will be to give a simple, unified proof of \eqref{loopresult} using the methods of orthogonal polynomials. Our results also apply to the \emph{generalised GUE}, leading to a 1-parameter extension of the asymptotics \eqref{loopresult}, which as far as we are aware have not appeared in either the mathematics or physics literature.

Like the ordinary GUE, the generalised GUE is defined on the set of $N \times N$ Hermitian matrices, but now the probability measure has the form
\begin{equation}\label{dP_gGUE}
dP(M_{N})=\frac{|\det(M_{N})|^{\lambda}}{\mathcal{Z}^{\mathrm{gGUE}}_N} \,\mathrm{exp}\left(-N\,\mathrm{Tr}(M_{N}^{2})\right)dM_{N},
\end{equation}
where we assume $\lambda>-1$ to ensure finiteness of the normalizing constant $\mathcal{Z}^{\mathrm{gGUE}}_N$:
\begin{equation}
\mathcal{Z}^{\mathrm{gGUE}}_N
=
\int_{\mathcal{H}_N} |\det(M_{N})|^{\lambda}\,\mathrm{exp}\left(-N\,\mathrm{Tr}(M_{N}^{2})\right)dM_{N}. \label{gguepf}
\end{equation}

This partition function depends implicitly on $\lambda$, although for simplicity of notation we do not emphasize it. Ensembles of the form \eqref{dP_gGUE}, with extra algebraic terms in the density, were studied extensively in the literature on matrix models under the name \textit{Gauss-Penner model}, see \cite{Deo,pen2} for details and applications. 

The main purpose of this paper is to prove the following

\begin{theorem}\label{Th1}
	Let $\mathbb{P}(M^{(\lambda)}_{N}>0)$ denote the probability that a random matrix from ensemble \eqref{dP_gGUE} is positive definite. Then for any fixed $\lambda>-1$, we have the asymptotic expansion as $N \to \infty$,
	\begin{equation}\label{asymp_PM}
		\begin{split}
&\log \mathbb{P}(M^{(\lambda)}_{N}>0) = -c_{1}N^{2}-\frac{\lambda \log(3)}{2}N+\left(c_{2}+\frac{\lambda^{2}}{4}\right)\log(N)+c_{3}\\
&+\frac{3\lambda^{2}}{4}\log(2)-\frac{\lambda^{2}}{2}\log(3)-\log\frac{G(\tfrac{3}{2})G(\tfrac{1}{2})G(\lambda+1)}{G(\tfrac{\lambda+3}{2})G(\tfrac{\lambda+1}{2})G(1)}+\mathcal{O}(N^{-1}),
\end{split}
	\end{equation}
	where $c_{1}, c_{2}$ and $c_{3}$ are the explicit constants defined above and $G(z)$ is the Barnes G function \cite[\S 5.17]{OLBC10}.
\end{theorem}
We note that in the case $\lambda=0$ we immediately recover the result \eqref{loopresult} of \cite{BEMN11} as a special case. We also mention the work \cite{B14} where the dependence of the leading term $c_{1}$ on growing $\lambda \sim N$ is investigated.

Figure \ref{Fig1} illustrates the accuracy of the asymptotic expansion 
\eqref{asymp_PM} for increasing $N$ and several values of 
$\lambda$. The comparison has been made with respect to brute force calculation of the 
Hankel determinant expression for the partition functions, see Appendix \ref{ApA}, which is quite time consuming and 
needs a large number of digits in {\sc Maple}.


\begin{figure}[!tb]
\centering
\includegraphics[width=70mm,height=62mm]{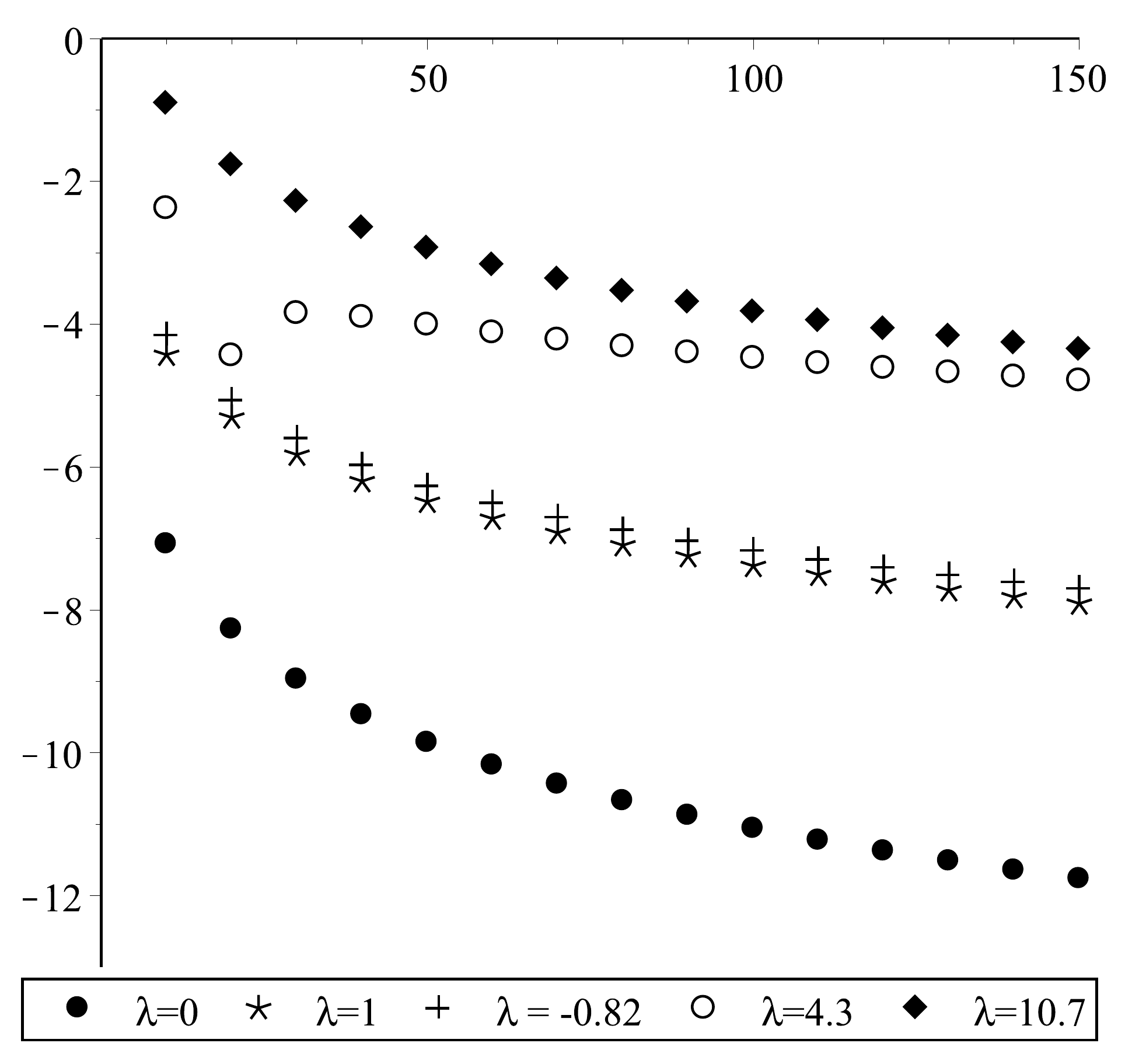}
\caption{Absolute errors (in $\log_{10}$ scale) as a function of $N$ and for different values of $\lambda$, taking 
all the terms in \eqref{asymp_PM} up to order $\mathcal{O}(1)$ (included).}
\label{Fig1}
\end{figure}

To prove Theorem \ref{Th1} we will study the partition function:
\begin{equation}
Z_{N}(s) = \int_{0}^{\infty}\ldots \int_{0}^{\infty}\prod_{j=1}^{N}w(x_j;\lambda,s)\,\prod_{1 \leq j < k \leq N}(x_{k}-x_{j})^{2}\,dx_{1}\ldots dx_{N} \label{zns}
\end{equation} 
where 
\begin{equation}\label{ws}
w(x;\lambda,s) = x^{\lambda}e^{-N(x+s(x^2-x))}.
\end{equation}

Note that this deformed weight interpolates between the classical Laguerre weight if $s=0$ and the generalized GUE if $s=1$. Diagonalizing $M_{N}$ in \eqref{dP_gGUE} and integrating out the eigenvectors (see \textit{e.g.} \cite{AGZ,Forrester,Mehta}) we see that
\begin{equation}
\begin{split}
\log\mathbb{P}(M^{(\lambda)}_{N}>0) &= \log\left(\frac{Z_{N}(1)}{Z_{N}^{\mathrm{gGUE}}}\right) \\
&= \int_{0}^{1}\frac{Z_{N}'(s)}{Z_{N}(s)}\,ds+\log Z_{N}(0)-\log Z^{\mathrm{gGUE}}_N.
\end{split} \label{s1}
\end{equation}
As the quantities $Z_{N}(0)=Z^{\textrm{LUE}}_{N}$ and $Z_{N}^{\mathrm{gGUE}}$ turn out to have explicit evaluations in terms of Gamma functions (see Lemmas \ref{lem:ApA} and \ref{lem:ApB}), our main task is to compute the integrand in \eqref{s1}. 
\begin{enumerate}
	\item We write $Z'_{N}(s)/Z_{N}(s)$ in terms of the recurrence coefficients 
	of a suitable family of 
	semiclassical Laguerre polynomials, orthogonal with respect to $w(x;\lambda,s)$ on $x \in [0,\infty)$.
	\item We compute the first terms in the asymptotic expansion of these recurrence coefficients 
	as $N\to\infty$, using the corresponding Riemann--Hilbert problem and the Deift--Zhou method of steepest descent.
	\item We show that such asymptotic expansions are uniform in $s\in[0,1]$ and we integrate term by term in \eqref{s1}.
\end{enumerate}

\section{Proof of Theorem \ref{Th1}}
Semi-classical Laguerre orthogonal polynomials (OPs): 
$\pi_{n,N}(x)=\pi_{n,N}(x;\lambda,s)$ are defined by the orthogonality
\begin{equation}\label{orthogonality}
\int_0^{\infty} \pi_{n,N}(x) x^k w(x;\lambda,s)dx=0, \qquad k=0,1,2,\ldots,n-1
\end{equation}
and the normalization
\begin{equation}\label{hnN}
\int_0^{\infty} \pi^2_{n,N}(x) w(x;\lambda,s)dx=h_{n,N}(\lambda,s)\neq 0, \qquad n\geq 0,
\end{equation}
where we write the weight function \eqref{ws} as
\begin{equation}
w(x;\lambda,s)=x^{\lambda}e^{-NV(x;s)}, \qquad \lambda>-1
\end{equation}
and $N>0$ is a real parameter. Here the potential is 
\begin{equation}\label{Vs}
V(x;s) = x+s(x^{2}-x),
\end{equation}
constructed in such a way that $V(x;0) = V(x) = x$ corresponds to the classical Laguerre OPs, while 
$V(x;1)=x^{2}$ is the potential that we are interested in. 

\begin{remark}
The deformation \eqref{Vs} follows a similar idea as the construction by Bleher and Its in \cite{BI}. The quantities considered here were also recently investigated in the complementary regime of fixed $N$ and large parameters by Clarkson and Jordaan \cite{CJ}. Part of this interest stems from the fact that the recurrence coefficients for semi-classical Laguerre polynomials, with weight $w(x;\lambda,t)=x^{\lambda}\exp(-x^2+tx)$, satisfy deformation equations (in $t$) that are closely related to the Painlev\'e IV differential equation \cite{BvA,FvAZ12,CJ}. Different aspects of this relationship were also studied in \cite{FW, WF12}.
\end{remark}

Since the weight function \eqref{ws} is positive and integrable on $[0,\infty)$ for $s\in[0,1]$, it follows from general theory \cite{Chihara, Ismail}, that the orthogonal polynomials $\pi_{n,N}(x)$ exist uniquely for all $n\geq 0$ and $s\in[0,1]$, 
and they satisfy $\deg\pi_{n,N}=n$. Furthermore, they are solutions of a three term recurrence relation (written in monic form):
\begin{equation}\label{TTRRs}
x\pi_{n,N}(x)=\pi_{n+1,N}(x)+\alpha_{n,N}\pi_{n,N}(x)+
\beta_{n,N} \pi_{n-1,N}(x),
\end{equation}
with initial data $\pi_{-1,N}(x)=0$, $\pi_{0,N}(x)=1$, and with recurrence coefficients $\alpha_{n,N}=\alpha_{n,N}(\lambda,s)$ and $\beta_{n,N}=\beta_{n,N}(\lambda,s)$.

\begin{remark}
For brevity, in the sequel, we will write $\alpha_n$ and $\beta_n$ instead of $\alpha_{n,N}$ and $\beta_{n,N}$, if no confusion arises, and similarly $h_n$ instead of $h_{n,N}$ in \eqref{hnN}. We will also use $\pi_n(x)$ instead of $\pi_{n,N}(x)$ for the $N$-dependent orthogonal polynomials and $w(x)$ instead of $w(x;\lambda,s)$ for the weight function \eqref{ws}.
\end{remark}

Next, we write $Z'_{N}(s)/Z_{N}(s)$, where the derivative is taken with respect to $s$, in terms of these recurrence 
coefficients. 


\begin{lemma}
\label{lem:def}
We have the following deformation equation
\begin{equation}
\frac{Z'_{N}(s)}{Z_{N}(s)} = \beta_{N}c_{N,\lambda}(s)-N^{2}\left[(1-3s)E_{N}+2sF_{N}\right], 
\label{def-formula}
\end{equation}
where
\begin{align}
c_{N,\lambda}(s) &:= N^2(3-s)+\lambda N\\ 
E_{N} &:= \beta_{N}(\alpha_{N}+\alpha_{N-1}) \label{EN}\\
F_{N} &:= \beta_{N}(\beta_{N+1}+\beta_N+\beta_{N-1}+\alpha_{N}^{2}+\alpha_{N}\alpha_{N-1}+\alpha_{N-1}^{2}) \label{FN}
\end{align}
\end{lemma}

\begin{proof}
Differentiating \eqref{zns} shows that
\begin{equation}
\frac{Z'_{N}(s)}{Z_{N}(s)} = -N\mathbb{E}_{V}\left[\sum_{j=1}^{N}(x_{j}^{2}-x_{j})\right] \label{linstat}
\end{equation}
where $\mathbb{E}_{V}$ denotes expectation with respect to the joint probability density function proportional to
\begin{equation}
\prod_{j=1}^{N}x_{j}^{\lambda}e^{-NV(x_{j},\lambda,s)}\prod_{1 \leq j < k \leq N}(x_{k}-x_{j})^{2} \label{pdfV}
\end{equation}
and $x_{j} \in [0,\infty)$, $j=1,\ldots,N$. The latter expectation can be written 
\begin{equation}
-N\mathbb{E}_{V}\left[\sum_{j=1}^{N}(x_{j}^{2}-x_{j})\right] = -N\int_{0}^{\infty}(x^{2}-x)\rho_{N}(x)\,dx \label{intdensity}
\end{equation}
where $\rho_{N}(x)$ is the so-called `one-point correlation function' or `eigenvalue density' corresponding to \eqref{pdfV}, see \textit{e.g.} \cite{AGZ,Mehta} for definitions and basic properties of this quantity. The equality \eqref{intdensity} appears frequently in random matrix theory, appropriate references include \cite[Eqn. 1.1.20 and Eqn. 1.1.41]{PS11}, see also \cite[Eqn. 1.8]{ZCD} where it was used for a similar purpose. In particular it is known that $\rho_{N}(x)$ can be computed explicitly by means of the Christoffel-Darboux formula:
\begin{equation}
\rho_{N}(x) = w(x)\,\frac{\pi_{N}'(x)\pi_{N-1}(x)-\pi_{N}(x)\pi_{N-1}'(x)}{h_{N-1}} \label{cd}
\end{equation}
Inserting \eqref{cd} into \eqref{intdensity} yields four different contributions which can all be written in terms of the recurrence coefficients $\alpha_{N}$ and $\beta_{N}$. One term vanishes due to
\begin{equation}
\int_{0}^{\infty}x\pi_{N-1}'(x)\pi_{N}(x)w(x)\,dx=0,
\end{equation}
a consequence of orthogonality. So \eqref{intdensity} can be decomposed as $I = I_{1}+I_{2}+I_{3}$, where
\begin{align}
I_{1} &:= \frac{N}{h_{N-1}}\int_{0}^{\infty}x^{2}\pi_{N-1}'(x)\pi_{N}(x)w(x)\,dx\\
I_{2} &:= -\frac{N}{h_{N-1}}\int_{0}^{\infty}x^{2}\pi_{N}'(x)\pi_{N-1}(x)w(x)\,dx\\
I_{3} &:= \frac{N}{h_{N-1}}\int_{0}^{\infty}x\pi_{N}'(x)\pi_{N-1}(x)w(x)\,dx
\end{align}
First observe that $I_{1} = N(N-1)\beta_{N}$ (as a consequence of $h_{N}/h_{N-1} = \beta_{N}$). An exercise in integration by parts shows that
\begin{align}
I_{2} &= N(N+1+\lambda)\beta_{N}-N^{2}(1-s)E_{N}-2sN^{2}F_{N}\\
I_{3} &= (1-s)N^{2}\beta_{N}+N^{2}2sE_{N}
\end{align}
where
\begin{align}
E_{N} &:= \frac{1}{h_{N-1}}\int_{0}^{\infty}\pi_{N}(x)\pi_{N-1}(x)x^{2}w(x)\,dx\\
F_{N} &:= \frac{1}{h_{N-1}}\int_{0}^{\infty}\pi_{N}(x)\pi_{N-1}(x)x^{3}w(x)\,dx
\end{align}
Combining all these terms yields \eqref{def-formula}. Finally the identities \eqref{EN} and \eqref{FN} follow from the three term recurrence relation \eqref{TTRRs}.
\end{proof}
The recurrence coefficients in Lemma \ref{lem:def} can be computed by solving the following coupled system of recurrence relations in the limit $N \to \infty$.
\begin{proposition}[String equations]
\label{prop:strings}
Let $n/N=q$ and $s\in[0,1]$, then the recurrence coefficients $\alpha_{n}$ and $\beta_{n}$ in \eqref{TTRRs} (omitting $N$ for brevity) satisfy 
\begin{equation}\label{string_Vs2}
\begin{aligned}
2s(\beta_{n+1}+\beta_{n}+\alpha_{n}^2)+(1-s)\alpha_{n}&=2q+\frac{\lambda+1}{N},\\
\beta_{n}\left(2s\alpha_{n}+1-s\right)\left(2s\alpha_{n-1}+1-s\right)
&=(2s\beta_{n}-q)\left(2s\beta_{n}-q-\frac{\lambda}{N}\right),
\end{aligned}
\end{equation}
with the values at $s=0$ corresponding to the (scaled with $N$) Laguerre 
polynomials:
$$
\alpha_{n}(0)=2q+\frac{\lambda+1}{N}, \qquad 
\beta_{n}(0)=q\left(q+\frac{\lambda}{N}\right).
$$
\end{proposition}

\begin{proof}
The string equations are known from \cite[Theorem 1.1]{BvA}, \cite{FvAZ12}
adapting the potential $V(x;s)=x^2-sx$ to the present one.
\end{proof}
We remark in passing that Boelen and Van Assche \cite{BvA} have shown that \eqref{string_Vs2} can be obtained from an asymmetric discrete Painlev\'e IV equation by a limiting process. 

To solve this system of equations asymptotically as $N\to\infty$, we exploit the following fact, the proof of which is postponed to the next section.
\begin{proposition}
\label{prop:1overN}
Let $n/N=q$. 
For any $\lambda >-1$,  the recurrence coefficients $\alpha_{n}$ and $\beta_{n}$ in \eqref{TTRRs} (omitting $N$ for brevity) have asymptotic expansions in inverse powers of $N$, as $N\to\infty$:
\begin{equation}\label{asympanbn}
\begin{aligned}
\alpha_n&=\alpha_{n}(q,\lambda,s)\sim \sum_{k=0}^{\infty} f_k(q,\lambda,s) N^{-k},\\
\beta_n&=\beta_{n}(q,\lambda,s)\sim \sum_{k=0}^{\infty} g_k(q,\lambda,s) N^{-k}.
\end{aligned}
\end{equation}
The coefficients $f_k(q,\lambda,s)$ and $g_k(q,\lambda,s)$ 
are real analytic functions of $s\in[0,1]$, and these expansions hold uniformly for $s$ in this interval.
\end{proposition}

With these ingredients in hand, we can now prove Theorem \ref{Th1}. We insert the expansions \eqref{asympanbn} into the recurrence \eqref{string_Vs2} and equate terms with equal powers of $N$. At leading order the solution consistent with the values at $s=0$ is
\begin{equation}
f_0=\frac{s-1+\Delta}{6s}, \qquad
g_0=\frac{(\Delta+s-1)\Delta}{72s^2},
\end{equation}
where $\Delta=\Delta(q,s)=\sqrt{s^2+24qs-2s+1}$. Next, we have
\begin{equation}
f_1=\frac{\lambda+1}{\Delta}, \qquad
g_1=\frac{(\Delta+s-1)\lambda}{12\Delta s}.
\end{equation}

Higher order corrections can be computed systematically in {\sc Maple}, but become quite cumbersome. If 
we substitute the terms up to order $\mathcal{O}(N^{-2})$ (included), with $n=N$ (so $q=1$) into the right hand side of 
\eqref{def-formula}, we get
\begin{equation}
\frac{Z'_{N}(s)}{Z_{N}(s)} = A(s)N^2+B(s)N+C(s)+\mathcal{O}(N^{-1}), \qquad N\to\infty ,
\label{def-formula_ABC}
\end{equation}
where
\begin{align}
A(s)&=\frac{\Delta^3(s+1)+s^4+34s^3-216s^2-34s-1}{432s^3},\\
B(s)&=\frac{\lambda(s^2-12s-1+(s+1)\Delta)}{24s^2},\\
C(s)&=\frac{\lambda^2(s+1)[s^2+6s+1+(s-1)\Delta]}{4s(s^2-10s+1)\Delta}\\
&-\frac{(s+1)^3\Delta+(s^2-1)(s^2+14s+1)}{12s(s^2-10s+1)\Delta^2},
\end{align}
Now integrating from $s=0$ to $s=1$, we get 
\begin{equation}\label{int-formula_ABC}
\begin{aligned}
\int_0^1 \frac{Z'_{N}(s)}{Z_{N}(s)} ds&=
N^2\int_0^1 A(s)ds+N\int_0^1 B(s)ds+\int_0^1 C(s)ds+
\mathcal{O}(N^{-1})\\
&=N^2\left(\frac{3}{4}-\frac{\log 6}{2}\right)
+N\left(\frac{1}{2}-\frac{\log 6}{2}\right)\lambda\\
&+\frac{\lambda^2\log (2/3)}{2}+\frac{\log 3}{8}-\frac{\log 2}{6}
+\mathcal{O}(N^{-1}).
\end{aligned}
\end{equation}
The integrals in \eqref{int-formula_ABC} are easily calculated in any computer algebra package. Combining \eqref{int-formula_ABC} with the known asymptotics for $\log Z_{N}(0)$ and $\log Z^{\mathrm{gGUE}}_{N}$ (see Lemmas \ref{lem:ApA} and \ref{lem:ApB} respectively) in \eqref{s1} completes the proof of Theorem \ref{Th1}. In the next section we prove Proposition \ref{prop:1overN}. 
\section{$1/N$ expansion for the recurrence coefficients}
The main purpose of this section is to justify the Ansatz \eqref{asympanbn} which we inserted into the string equations. This is based on the fact that the recurrence coefficients can be computed in terms of the solution of an appropriate Riemann-Hilbert problem (RHP). Then their asymptotics can be analysed very precisely using the Deift--Zhou method of steepest descent. 

\subsection{Equilibrium measure}
In the steepest descent analysis, a key role is played by the equilibrium measure $d\mu_V$, which minimizes the 
logarithmic energy
\begin{equation}
E(\nu)=\iint \log\frac{1}{|x-y|}d\nu(x)d\nu(y)+\int V(x;s)d\nu(x),
\end{equation}
over all probability measures supported on $[0,\infty)$, where the external field $V(x;s)$ is given by \eqref{Vs}. 
Such a problem has a unique solution, say $d\mu_V(x;s)$, since $w(x;\lambda,s)=x^{\lambda}e^{-NV(x;s)}$ is an admissible 
weight function in the sense of Saff and Totik \cite[Def. 1.1]{ST}. Moreover, we have the variational equations
\begin{equation}\label{variational}
\begin{aligned}
-g_+(x;s)-g_-(x;s)+V(x;s)&=\ell, \qquad x\in\textrm{supp}\,\mu_V(x;s),\\
-g_+(x;s)-g_-(x;s)+V(x;s)&\geq \ell, \qquad x\,\,\textrm{a.e. in}\, [0,\infty),\\
\end{aligned}
\end{equation}
where $g(z;s)=\int \log(z-x)d\mu_V(x;s)$ is analytic for $z\in\mathbb{C}\setminus \mathbb{R}$, and $g_{\pm}(x;s)$ indicate the boundary values
$$
g_{\pm}(x;s)=\lim_{\varepsilon\to 0^+} g(x\pm i\varepsilon;s), \qquad x\in\mathbb{R}.
$$

In our case, the support and density of the equilibrium measure can be worked out explicitly:


\begin{lemma}\label{lemma_eq}
Let $s\in[0,1]$, the equilibrium measure corresponding to the weight function 
$w(x;\lambda,s)=x^{\lambda} e^{-NV(x;s)}$, with $V(x;s)$ given by \eqref{Vs}, is supported on 
the interval $(0,c)$, where
\begin{equation}\label{c}
c=\frac{s-1+\sqrt{s^2+22s+1}}{3s}.
\end{equation}
If we write $d\mu_V(x;s)=\psi_V(x;s)dx$, the density is given by
\begin{equation}
\psi_V(x;s)=-\frac{1}{\pi}(ax+b)\sqrt{\frac{c-x}{x}}, 
\end{equation}
with 
\begin{equation}\label{ab}
a=-s, \qquad b=\frac{2s-2-\sqrt{s^2+22s+1}}{6}.
\end{equation}
\end{lemma}

\begin{proof}
The potential $V(x;s)=x+s(x^2-x)$ is convex for any $s\in[0,1]$, hence following the general theory, see for instance the monograph of Saff and Totik \cite[Chapter IV, Theorem 1.11]{ST}, the equilibrium measure is supported on a single interval, say 
$[0,c]$. If such  equilibrium measure is $d\mu_V(x;s)=\psi_V(x;s)dx$, the function
$$
\omega(z;s)=\int_0^c \frac{\psi_V(x;s)}{z-x}dx
$$
satisfies 
\begin{equation}\label{eq_resolvent}
\begin{aligned}
\omega(z;s)&=\frac{1}{z}+\mathcal{O}(z^{-2}), \qquad z\to\infty,\\
\omega_+(x;s)+\omega_-(x;s)&=V'(x;s), \qquad x\in (0,c),
\end{aligned}
\end{equation}
the second identity being a consequence of the variational equations \eqref{variational}.

Consequently, we look for $\omega(z;s)$ of the form
$$
\omega(z;s)=\frac{V'(z;s)}{2}+(az+b)(z-c)^{1/2},
$$
with a branch cut on $[0,c]$. The first equation in \eqref{eq_resolvent} gives the coefficients $a$, $b$ and $c$ 
in \eqref{c} and \eqref{ab}.
\end{proof}

We observe that the form of the equilibrium measure is uniform in $s\in[0,1]$. A straightforward calculation from \eqref{c} shows that $c=c(s)$ is a decreasing function for $s\in[0,1]$, and $c\in[\tfrac{2}{3}\sqrt{6},4]$. Similarly, the extra zero of the density is $-b/a$, where $a$ and $b$ are given in \eqref{ab}, and it increasing with $s$ from $-\infty$ to $-\tfrac{2}{3}\sqrt{6}$. Since this extra zero is bounded away from the support of the equilibrium measure when $s\in[0,1]$, no critical transitions take place. This fact will be crucial in the calculation of the asymptotic expansions below.

\subsection{RH problem}

Following the original idea of Fokas, Its and Kitaev \cite{FIK} in this context, the (monic) semiclassical 
Laguerre polynomials $\pi_n(x)$ are the $(1,1)$ entry of a $2\times 2$ matrix 
$Y(z)=Y_n(z;\lambda,s):\mathbb{C}\setminus[0,\infty)\mapsto\mathbb{C}^{2\times 2}$ that satisfies the following RH problem:
\begin{enumerate}
\item $Y(z)$ is analytic in $\mathbb{C}\setminus[0,\infty)$.
\item On $(0,\infty)$, oriented from left to right, the boundary values of $Y$ satisfy
$$
Y_+(x)=Y_-(x)
\begin{pmatrix}
1 & x^{\lambda} e^{-NV(x;s)}\\
0 & 1
\end{pmatrix},
$$
where $Y_{\pm}(x)=\lim_{\varepsilon\to 0^+} Y(x\pm i\varepsilon)$, taken entrywise, indicates the boundary values from above and below the real axis respectively.
\item As $z\to\infty$, we have
\begin{equation}\label{asympY}
Y(z)=\left(I+\frac{Y_1}{z}+\frac{Y_2}{z^2}+\mathcal{O}\left(\frac{1}{z^3}\right)\right)
\begin{pmatrix}
z^n & 0\\
0 & z^{-n}
\end{pmatrix}
\end{equation}
\item As $z\to 0$, we have
\begin{equation}\label{asympY0}
Y(z)=
\begin{cases}
\begin{pmatrix}
\mathcal{O}(1) & \mathcal{O}(z^{\lambda})\\
\mathcal{O}(1) & \mathcal{O}(z^{\lambda})
\end{pmatrix},&  \qquad \lambda<0,\\
\begin{pmatrix}
\mathcal{O}(1) & \mathcal{O}(\log z)\\
\mathcal{O}(1) & \mathcal{O}(\log z)
\end{pmatrix},&  \qquad \lambda=0,\\
\begin{pmatrix}
\mathcal{O}(1) & \mathcal{O}(1)\\
\mathcal{O}(1) & \mathcal{O}(1)
\end{pmatrix},&  \qquad \lambda>0.
\end{cases}
\end{equation}
\end{enumerate}

It is known \cite{Deift} that the recurrence coefficients in 
\eqref{TTRRs} can be written as follows:
\begin{equation}\label{anbnY}
\alpha_{n,N}=\frac{(Y_2)_{12}}{(Y_1)_{12}}-(Y_1)_{22}, \qquad
\beta_{n,N}=(Y_1)_{12}(Y_1)_{21}
\end{equation}
where $Y_1$ and $Y_2$ are the matrices that appear in the asymptotic expansion \eqref{asympY}, 
see \cite[\S 3.2]{Ble} or \cite{Deift}.

\subsection{Steepest descent}
The steepest descent method of Deift and Zhou consists of a series of transformations that lead to a final RH problem that can be solved asymptotically as $N\to\infty$, uniformly in $z$ in the complex plane. Since we are only using the steepest descent method in order to 
prove existence of an asymptotic expansion in powers of $1/N$ for the recurrence coefficients, and not to obtain the details of the coefficients therein, the presentation will be quite brief. We refer the reader to the work of Vanlessen \cite{Vanlessen}, or 
Zhao et al. \cite{ZCD} for a more detailed explanation in a similar setting.

The basic steps in this case are the following:
\begin{equation}\label{steepest}
Y\mapsto T \mapsto S\mapsto R.
\end{equation}

The first step $Y\mapsto T$ is a normalization at infinity:
\begin{equation}\label{TY}
T(z)=
\begin{pmatrix}
e^{-N\ell/2} & 0\\
0 & e^{N\ell/2}
\end{pmatrix}
Y(z)
\begin{pmatrix}
e^{-N(g(z;s)-\ell/2)} & 0\\
0 & e^{N(g(z;s)-\ell/2)}
\end{pmatrix},
\end{equation}
where $\ell$ is a constant in $N$ (Lagrange multiplier of the equilibrium problem), and  $g$ is the logarithmic transform 
of the equilibrium measure:
\begin{equation}\label{gfunction}
g(z;s)=\int_0^{c} \log(z-x)d\mu_V(x;s),
\end{equation}
which is analytic in $\mathbb{C}\setminus (-\infty,c]$, with $c$ given by \eqref{c}, and as $z\to\infty$ satisfies
\begin{equation}\label{asympg}
g(z;s)=\log z-\frac{\mu_{1}(s)}{z}-\frac{\mu_{2}(s)}{2z^2}+\mathcal{O}(z^{-3}),
\end{equation}
where $\mu_{k}(s)=\int_0^{c} x^k d\mu_V(x;s)$, $k\geq 1$, are the moments of the measure $d\mu_V$, that can be 
computed explicitly. As a consequence, we have the expansion
\begin{equation}\label{asympexpg}
\begin{aligned}
e^{Ng(z;s)\sigma_3}
&=
\begin{pmatrix}
e^{Ng(z;s)} & 0\\
0 & e^{-Ng(z;s)}
\end{pmatrix}\\
&=
\begin{pmatrix}
z^N & 0\\
0 & z^{-N}
\end{pmatrix}
\left(
I+\frac{G_1}{z}+\frac{G_2}{z}+\mathcal{O}(z^{-3})
\right),
\end{aligned}
\end{equation}
as $z\to\infty$, where $G_1$ and $G_2$ are diagonal matrices (and dependent of $s$ and $N$). 

The second step $T(z)\mapsto S(z)$ deforms the jump contours by opening a lens around the interval $[0,c]$. This step does not 
make any change away from a small neighbourhood of $[0,c]$, and since we will be using 
information as $z\to\infty$ for the recurrence coefficients, see \eqref{anbnY}, we can replace $T(z)=S(z)$. 

The final step, $S(z)\mapsto R(z)$ involves both a global unimodular parametrix $P^{(\infty)}(z)$, away from the endpoints $z=0$ and $z=c$, and two unimodular 
local parametrices, $P_{\textrm{Airy}}(z)$ and $P_{\textrm{Bessel}}(z)$ built out of Airy functions in a neighbourhood of $z=c$ and Bessel functions in a neighbourhood of $z=0$. Then we construct 
$$
R(z)=
\begin{cases}
S(z)[P^{(\infty)}]^{-1}(z), & \qquad z\in\mathbb{C}\setminus \overline{D_{\delta}(0)}\cup \overline{D_{\delta}(c)},\\
S(z)[P_{\textrm{Airy}}]^{-1}(z), & \qquad z\in D_{\delta}(c),\\
S(z)[P_{\textrm{Bessel}}]^{-1}(z), & \qquad z\in D_{\delta}(0),\\
\end{cases}
$$
where $D_{\delta}(0)$ and $D_{\delta}(c)$ are discs of fixed radius $\delta>0$ around $z=0$ and $z=c$ respectively. 
The RH problem for $R(z)$ can be solved iteratively, since $R$ is normalized at infinity and 
all jumps are close to the identity, see \cite[\S 11]{Ble} or \cite{Deift}. The consequence is an asymptotic expansion of the form:
\begin{equation}\label{asympRN}
R(z)\sim\sum_{m=0}^{\infty}\frac{R^{(m)}(z)}{N^m}, \qquad N\to\infty,
\end{equation}
uniformly in $z$ away from a contour $\Sigma_R$ around the interval $[0,c]$, see 
\cite[Chapter 7]{Deift} or \cite{ZCD}. It is at this stage that the uniform form of the equilibrium 
measure with respect to $s$ is fundamentally important, since the parametrices 
depend on $s$ but they have the same structure for $s\in[0,1]$, i.e. \eqref{asympRN} holds uniformly with respect to $s\in[0,1]$. 

In addition, $R$ has an asymptotic expansion as $z\to\infty$, that we write
\begin{equation}\label{asympRz}
R(z)\sim I+\sum_{k=1}^{\infty}\frac{R_k}{z^m},
\end{equation}
and combining \eqref{asympRz} with \eqref{asympRN}, each coefficient $R_k$ can be expanded 
asymptotically in inverse powers of $N$.

Away from the interval $[0,c]$, we write $T(z)=S(z)=R(z)P^{(\infty)}(z)$ and replace this in \eqref{TY}:
\begin{equation}\label{YR}
Y(z)=\begin{pmatrix}
e^{N\ell/2} & 0\\
0 & e^{-N\ell/2}
\end{pmatrix}
R(z)P^{(\infty)}(z)
\begin{pmatrix}
e^{N(g(z;s)-\ell/2)} & 0\\
0 & e^{-N(g(z;s)-\ell/2)}
\end{pmatrix}.
\end{equation}

The global parametrix 
$P^{(\infty)}$ satisfies a RH problem analogous to the one presented in \cite[Section 3.5]{Vanlessen}, 
but on $[0,c]$ instead of $[0,1]$. Making the corresponding changes, we have
\begin{equation}\label{asympPinf}
P^{(\infty)}(z)
=
I+\frac{P_1^{(\infty)}}{z}+\frac{P_2^{(\infty)}}{z^2}+\mathcal{O}(z^{-3}), 
\end{equation}  
as $z\to\infty$, with some matrices $P_1^{(\infty)}$ and $P_2^{(\infty)}$ that can be computed explicitly, but whose precise 
form is not relevant in the present discussion. Using \eqref{asympRz} in \eqref{YR} and identifying terms, we obtain
\begin{equation}\label{Y1Y2R}
\begin{aligned}
Y_1
&=
e^{\frac{N\ell\sigma_3}{2}}
(P_1^{(\infty)}+G_1+R_1)
e^{-\frac{N\ell\sigma_3}{2}}\\
Y_2
&=
e^{\frac{N\ell\sigma_3}{2}}
(P_2^{(\infty)}+G_2+R_2+R_1P_1^{(\infty)}+(P_1^{(\infty)}+R_1)G_1)
e^{-\frac{N\ell\sigma_3}{2}}
\end{aligned}
\end{equation}

From this, we can obtain an expression for the recurrence coefficients in terms of all the 
matrices involved. The terms in the expansion of $P^{(\infty)}$ are independent of $N$, and the $G_k$ coefficients in \eqref{asympexpg} contain only integer powers of $N$. This result, together with \eqref{Y1Y2R}, gives asymptotic expansions 
in powers of $1/N$ for the recurrence coefficients $\alpha_{N,N}$ and $\beta_{N,N}$, as desired.

Finally, we note that this result also applies to $\alpha_{n,N}$ and $\beta_{n,N}$ with $n/N=q$, which is needed in the string equations \eqref{string_Vs2}. We can rewrite
$$
NV(x;s)=n\frac{V(x;s)}{q}, \qquad q=\frac{n}{N},
$$
and work with the potential $V(x;s)/q$ throughout. Since $q$ will be close to $1$ when both $n$ and $N$ are large, and all quantities depend 
analytically on $q$ (in particular the equilibrium measure in Lemma \ref{lemma_eq}), we get the same kind of asymptotic expansions in the steepest descent method.

\section*{Acknowledgements}
A. D. acknowledges financial support from projects MTM2012-36732-C03-01 and MTM2012-34787 from the Spanish Ministry of Economy and Competitivity. N. J. S. acknowledges financial support from a Leverhulme Trust Early Career Fellowship ECF-2014-309. The authors would like to thank the anonymous referees for a number of useful remarks and corrections that were added in the revised version.

\appendix
\section{Asymptotic expansions for LUE and gGUE partition functions}
\label{ApA}
\begin{lemma}
\label{lem:ApA}
The partition function of the Laguerre Unitary Ensemble:
\begin{equation}\label{ZNLUE}
Z^{\mathrm{LUE}}_{N}
=
\int_0^{\infty}\cdots\int_0^{\infty}\prod_{j=1}^N x_j^\lambda e^{-Nx_j}\prod_{1\leq j<k\leq N}(x_k-x_j)^2\, dx_{1}\ldots dx_{N},
\end{equation}
with $\lambda>-1$, can be written as
\begin{equation}\label{ZNLUE2}
Z^{\mathrm{LUE}}_{N}
=
N^{-N(N+\lambda)}\prod_{j=1}^N \Gamma(j+1)\Gamma(j+\lambda),
\end{equation}
and as $N\to\infty$ we have
\begin{equation}\label{ZNLUE_asymp}
\begin{aligned}
\log Z^{\mathrm{LUE}}_{N}
&=
-\frac{3}{2}N^2+N\log N+\left(\log(2\pi)-1-\lambda\right)N+\frac{3\lambda^2+2}{6}\log N\\
&+\frac{1+3(\lambda+1)\log(2\pi)}{6}-2\log A-\log G(\lambda+1)\\
&+\frac{2\lambda^3-\lambda+1}{12N}+\mathcal{O}(N^{-2}),
\end{aligned}
\end{equation}
where $G$ is the Barnes $G$-function, see \cite[\S 5.17]{OLBC10}, and 
\begin{equation}\label{A}
A=\exp\left(\frac{1}{12}-\zeta'(-1)\right)
\end{equation}
is the Glaisher--Kinkelin constant, $A=1.28242 71291 \ldots$
\end{lemma}
\begin{proof}
The explicit formula \eqref{ZNLUE2} is a consequence of the fact that \eqref{ZNLUE} can be 
written as a Selberg integral. See \cite[Theorem 2.5.8, Corollary 2.5.9]{AGZ}, and also 
\cite{ZCD} and the monograph by Mehta \cite{Mehta}. Alternatively, one can use the fact, see \cite[\S 18]{Ble}, that
$$
Z^{\mathrm{LUE}}_{N}=N!\prod_{j=0}^{N-1} h^{\textrm{L}}_{j}
=
N!\prod_{j=0}^{N-1}N^{-2j-\lambda-1}\Gamma(j+1)\Gamma(j+\lambda+1),
$$
in terms of the normalizing constants of (scaled and monic) Laguerre polynomials.

Next, we rewrite \eqref{ZNLUE2} as follows:
\begin{equation}\label{ZNLUE_G}
Z^{\mathrm{LUE}}_{N} = N^{-N(N+\lambda)}\frac{G(N+2)G(N+\lambda+1)}{G(2)G(\lambda+1)}
\end{equation}
again in terms of the Barnes $G$-function. 
This function has a known asymptotic expansion:
\begin{equation}\label{asympG}
\begin{aligned}
\log G\left(z+1\right)
&\sim\tfrac{1}{4}z^{2}+z\log\Gamma(z+1)-\left(\tfrac{1}{2}z(z+1)+\tfrac{1}{12}\right)\log z\\
&-\log A+\sum_{k=1}^{\infty}\frac{B_{2k+2}}{2k(2k+1)(2k+2)z^{2k}}, 
\end{aligned}
\end{equation}
as $z\to\infty$ with $|\textrm{arg}\, z|<\pi$, see for example \cite[5.17.5]{OLBC10}. Here $B_{2k+2}$ are Bernoulli numbers. Replacing this asymptotic expansion in
\eqref{ZNLUE_G} and using {\sc Maple}, we obtain \eqref{ZNLUE_asymp}.
\end{proof}


Next, we consider the generalised GUE partition function:
\begin{equation}\label{pf_gGUE}
Z_{N}^{\mathrm{gGUE}} := \int_{\mathbb{R}^{N}}\prod_{j=1}^{N}|x_{j}|^{\lambda}e^{-Nx_{j}^{2}}\prod_{1 \leq k < j \leq N}(x_{k}-x_{j})^{2}\,dx_{1}\ldots dx_{N}
\end{equation}

\begin{lemma}
\label{lem:ApB}
For fixed $\lambda>-1$, the partition function \eqref{pf_gGUE} can be written as 
\begin{align}\label{pf_gGUE2}
\nonumber
Z^{\mathrm{gGUE}}_{N} &= (2N)^{-N^{2}/2}(2\pi)^{N/2}N^{-\lambda N/2}\prod_{j=1}^{N}
\frac{\Gamma(\tfrac{\lambda+1}{2}+\lfloor \tfrac{j}{2} \rfloor)}{\Gamma(\tfrac{1}{2}+\lfloor \tfrac{j}{2} \rfloor)}j!\\
&=\nonumber(2N)^{-N^{2}/2}(2\pi)^{N/2}N^{-\lambda N/2} \frac{G(\tfrac{3}{2})G(\tfrac{1}{2})}{G(\tfrac{\lambda+3}{2})G(\tfrac{\lambda+1}{2})}\\
&\times\frac{G(N+2)G(\tfrac{\lambda+N+3}{2})G(\tfrac{\lambda+N+1}{2})}{G(\tfrac{N+3}{2})G(\tfrac{N+1}{2})}
\end{align}
where $G$ is the Barnes $G$-function. Here $\lfloor j/2 \rfloor$ denotes the largest integer less than or equal to $j/2$, and 
we assumed that $N$ is even for simplicity. As $N\to\infty$, we have
\begin{equation}\label{ZNgGUE_asymp}
\begin{aligned}
\log Z^{\mathrm{gGUE}}_{N}
&=
\left(-\frac{3}{4}-\frac{\log 2}{2}\right)N^{2}
+N\log(N)\\
&+\left(\log(2\pi)-\frac{\lambda(1+\log(2))+2}{2}\right)N\\
&+\frac{3\lambda^{2}+5}{12}\log(N)
+c_{0}+\frac{c_{1}}{N}+\mathcal{O}(N^{-2}),
\end{aligned}
\end{equation}
where $c_{0}$ and $c_{1}$ are explicit constants
\begin{align}
\nonumber
c_{0} &= \frac{1-3\lambda^{2} \log(2)-12 \log A+6(\lambda+1)\log(2\pi)}{12}+
\log \frac{G(\tfrac{3}{2})G(\tfrac{1}{2})}{G(\tfrac{\lambda+3}{2})G(\tfrac{\lambda+1}{2})}\\ 
c_{1} &= \frac{\lambda^3+\lambda+1}{12}.
\end{align}
\end{lemma}
\begin{remark}
We point out that the partition function \eqref{pf_gGUE} admits a natural generalization involving a fixed number of spectral singularities (of Fisher-Hartwig type) in the integrand. The relevant asymptotics in that case were calculated in \cite{Kra07}. 
\end{remark}
\begin{proof}
The first equality in \eqref{pf_gGUE2} was obtained by Mehta and Normand in \cite{MN}. For completeness we reproduce their derivation here. The Heine identity 
\begin{equation}
Z_N(\lambda)=N!\, D_N(\lambda), \qquad D_N(\lambda)= \det\left[\mu_{j+k}\right]_{j,k=0}^{N-1}, \label{Heine}
\end{equation}
allows us to write the partition function in terms of the Hankel determinant, which is constructed 
with the moments of the weight function:
\begin{equation}
\mu_k=\mu_k(\lambda)=\int_0^{\infty} x^k x^{\lambda}e^{-x^{2}}\, dx, \qquad k\geq 0.
\end{equation}
Thus, the partition function \eqref{pf_gGUE} can be written as
\begin{equation}
\begin{split}
Z^{\mathrm{gGUE}}_{N} &= c^{(\lambda)}_{N}\,\det\bigg\{\int_{\mathbb{R}}\,x^{i+j}|x|^{\lambda}e^{-x^{2}}\,dx\bigg\}_{i,j=0}^{N-1} = c^{(\lambda)}_{N}\,\det\bigg\{\Phi_{i,j}\bigg\}_{i,j=0}^{N-1} \label{singledet}
\end{split}
\end{equation}
where $c^{(\lambda)}_{N} = N^{-N(N+\lambda)/2}N!$ and $\Phi_{i,j} = \Gamma((\lambda+1+i+j)/2)$ if $i+j$ is even and $\Phi_{i,j}=0$ if $i+j$ is odd. This determinant has a `checkerboard structure' of zeros and by elementary row and column manipulations, it can be arranged so that all $\Phi_{i,j}$ with purely even indices appear in the top-left block and $\Phi_{i,j}$ with odd indices in the bottom-right. This allows us to write \eqref{singledet} as a product
\begin{equation}
Z^{\mathrm{gGUE}}_{N} = c^{(\lambda)}_{N}\,\det\{\Phi_{2i,2j}\}_{i,j=0}^{\lfloor(N-1)/2\rfloor}\det\{\Phi_{2i+1,2j+1}\}_{i,j=0}^{\lfloor (N-2)/2\rfloor}. \label{gGUEprod}
\end{equation}
The latter determinants can be computed from the simple fact that for generic $z \in \mathbb{C}$ we have
\begin{equation}
\det\{\Gamma(z+i+j)\}_{i,j=0}^{M} = \prod_{j=0}^{M}j!\Gamma(z+j) \label{gammadet}
\end{equation}
which is a simple exercise to prove from, say the classical Laplace expansion of the determinant. Applying \eqref{gammadet} to \eqref{gGUEprod} shows that
\begin{equation}
\frac{Z^{\mathrm{gGUE}}_{N}}{Z^{\mathrm{GUE}}_{N}} = N^{-\lambda N/2}\prod_{j=1}^{N}\frac{\Gamma(\tfrac{\lambda+1}{2}+\lfloor \tfrac{j}{2} \rfloor)}{\Gamma(\tfrac{1}{2}+\lfloor \tfrac{j}{2} \rfloor)} \label{ratioGUEs}
\end{equation}
where we used that the left-hand side must equal $1$ when $\lambda=0$. The first equality in \eqref{pf_gGUE2} now follows from \eqref{ratioGUEs} and the well-known formula for $Z^{\mathrm{GUE}}_{N} := Z^{\mathrm{gGUE}}_{N}|_{\lambda=0}$ (see \textit{e.g.} \cite{Mehta}). The second equality in \eqref{pf_gGUE2} and the asymptotics follow from the general properties and corresponding asymptotic expansion \eqref{asympG} of the Barnes $G$-function.

\end{proof}


\begin{thebibliography}{99}
	
\bibitem{AAE} {\sc A. Aazami, R. Easther.} Cosmology from random multifield potentials.
{\it J. Cosmol. Astropart. Phys. 3 (2006), 013, 17.}

\bibitem{pen2} {\sc G. \'{A}lvarez, L. Mart{\'{\i}}nez Alonso, E. Medina.} 
Partition functions and the continuum limit in {P}enner matrix models. {\it J. Phys. A 47, 31 (2014), 315205, 29.}

\bibitem{AGZ} {\sc G. W. Anderson, A. Guionnet, O. Zeitouni.} An Introduction to Random Matrices. 
{\it Cambridge University Press, 2011.}

\bibitem{BG} {\sc G. Ben Arous, A. Guionnet.} Large deviations for {W}igner's law and {V}oiculescu's
	non-commutative entropy. {\it Probab. Theory Related Fields, 108, 4 (1997), 517--542.}

\bibitem{BCFJK} {\sc M. Bhargava, J. E. Cremona,  T. A. Fisher, N. G. Jones, J. P. Keating.} What is the probability that a random integral quadratic form in {$n$} variables has an integral zero? {\it Int. Math. Res. Not. IMRN 12 (2016), 3828--3848}.

\bibitem{Ble}{\sc P. M. Bleher.} Lectures on Random Matrix Models. The Riemann-Hilbert Approach. {\it In ``Random Matrices, Random Processes and Integrable Systems", CRM Series in Mathematical Physics (John Harnad, ed.), 251-349. Springer, 2011.}



\bibitem{BvA} {\sc L. Boelen, W. van Assche.} Discrete Painlev\'e equations for recurrence coefficients 
of semiclassical {L}aguerre polynomials. {\it Proc. Amer. Math. Soc. 138 (2010), 1317--1331.}

\bibitem{BG13}{\sc G. Borot, A. Guionnet.} Asymptotic expansion of $\beta$ matrix models in the one-cut regime. {\it Commun. Math. Phys., 317, 2, (2013) 447--483.}

\bibitem{BI}{\sc P. M. Bleher, A. R. Its.} Asymptotics of the partition function of a random matrix model. {\it Ann. Inst. Fourier (Grenoble), 55, no. 6, (2005) 1943--2000.}

\bibitem{BEMN11} {\sc G. Borot, B. Eynard, S. N. Majumdar, C. Nadal.} 
Large deviations of the maximal eigenvalue of random matrices. {\it J. Stat. Mech. Theory Exp. 11 (2011), P11024, 56.}


\bibitem{B14} {\sc M. Bouali.} Density of {P}ositive {E}igenvalues of the {G}eneralized {G}aussian {U}nitary {E}nsemble. 
{\it https://arxiv.org/abs/1409.0103}

\bibitem{CGG} {\sc A. Cavagna, J. P. Garrahan, I. Giardina.} 
Index {D}istribution of {R}andom {M}atrices with an {A}pplication to {D}isordered {S}ystems. 
{\it Phys. Rev. B, 61 (2000), 3690.}


\bibitem{Chihara} {\sc T. S. Chihara.} An Introduction to Orthogonal Polynomials. {\it Dover Publications, 2011.}

\bibitem{CJ} {\sc P. A. Clarkson, K. Jordaan.} The relationship between semi-classical Laguerre polynomials and the fourth Painlev\'e equation. {\it Const. Approx. 39, 1 (2014), 223--254.}


\bibitem{CK} {\sc T. Claeys, A. B. J. Kuijlaars.} Universality in random matrix ensembles when the soft edge meets the hard edge. {\it Contemporary Mathematics, 458 (2006), 265--280.}

\bibitem{CKV} {\sc T. Claeys, A. B. J. Kuijlaars, M. Vanlessen.} Multi-critical unitary random matrix ensembles and the general Painlev\'e II equation. {\it Ann. Math. 167 (2008), 601--642.}

\bibitem{DM} {\sc D. S. Dean, S. N. Majumdar.} Extreme Value Statistics of Eigenvalues of Gaussian Random Matrices. {\it Phys. Rev. E, 77, 041108 (2008).}

\bibitem{Deift} {\sc P. Deift.} Orthogonal polynomials and Random Matrices. The Riemann--Hilbert 
Approach. Volume 3 of {\it Courant Institute of Mathematical Sciences Lecture Notes. 
American Mathematical Society, 2000.}

\bibitem{DKMVZ} {\sc P. Deift, T. Kriecherbauer, K.T-R McLaughlin, S. Venakides, X. Zhou.} Uniform asymptotics
for polynomials orthogonal with respect to varying exponential weights and applications
to universality questions in random matrix theory. {\it Commun. Pure Appl. Math. 52,
no. 11, (1999), 1335--1425.}

\bibitem{Deo} {\sc N. Deo.} Glassy random matrix models. {\it Phys. Rev. E (3), 65, 5 (2002), 056115, 10.}

\bibitem{FW} {\sc P. J. Forrester and N. S. Witte.} Application of the $\tau$-function {T}heory of {P}ainlev\'e {E}quations to {R}andom {M}atrices: {PIV}, {PII} and the {GUE}. {\it Commun. Math. Phys. 219, 2 (2001), 357-398.}

\bibitem{FvAZ12} {\sc G. Filipuk, W. Van Assche, L. Zhang.} 
The recurrence coefficients of semi-classical {L}aguerre polynomials and the fourth {P}ainlev\'e equation. 
{\it J. Phys. A 45, 20 (2012), 205201, 13.}

\bibitem{FIKN} {\sc A. Fokas, A. R. Its, A. A. Kapaev, V. Yu. Novokshenov.} Painlev\'{e} Transcendents. 
The Riemann--Hilbert Approach.  {\it AMS, 2006.}

\bibitem{FIK} {\sc A. S. Fokas, A. R. Its, A. V. Kitaev.} 
The isomonodromy approach to matrix models in 2D quantum gravity. {\it Comm. Math. Phys. 147, 2 (1992), 396--430.}

\bibitem{Forrester} {\sc P. Forrester.} Log-Gases and Random Matrices. Volume 34 of {\em The London Mathematical Society 
Monographs Series. Princeton University Press, 2010.}

\bibitem{Fyodorov} {\sc Y. V. Fyodorov.} 
Complexity of {R}andom {E}nergy {L}andscapes, {G}lass {T}ransition and {A}bsolute {V}alue of {S}pectral {D}eterminant of {R}andom {M}atrices. {\it Phys. Rev. Lett. \textbf{92}, 240601 (2004). }

\bibitem{Ismail} {\sc M. E. H. Ismail.} Classical and Quantum Orthogonal Polynomials in One Variable.  {\it Cambridge University Press, 2009.}

\bibitem{Kra07} {\sc I. Krasovsky.} Correlations of the characteristic polynomial in the {G}aussian unitary ensemble or a singular {H}ankel determinant. {\it Duke Math. J. 139 (2007) 581-619.}

\bibitem{KMcLVaV} {\sc A. B. J. Kuijlaars, K. T.-R. McLaughlin, W. van Assche, M. Vanlessen.} The Riemann-Hilbert approach to strong asymptotics for orthogonal polynomials on $[-1,1]$. {\it Adv. Math. 188 (2004), 337--398.}

\bibitem{Mehta} {\sc M. L. Mehta.} Random Matrices. {\it Academic Press, 2004.}

\bibitem{MN} {\sc M. L. Mehta, J-M. Normand.} Probability density of the determinant of a random Hermitian matrix. {\it J. Phys. A: Math. Gen. 31 (1998) 5377-5391.}

\bibitem{DLMF}
NIST Digital Library of Mathematical Functions. 
http://dlmf.nist.gov/, Release 1.0.6 of 2013-05-06. 
Online companion to \cite{OLBC10}

\bibitem{OLBC10}
{\sc F. W. J. Olver, D. W. Lozier, R. F. Boisvert, C. W. Clark, eds.}
NIST Handbook of Mathematical Functions. 
Cambridge University Press, New York, NY, 2010. 
Print companion to \cite{DLMF}.

\bibitem{PS11}{\sc L. Pastur and M. Shcherbina.}
Eigenvalue distribution of large random matrices.
{\it American Mathematical Society, Providence, RI, 2011}

\bibitem{ST} {\sc E. B. Saff, V. Totik.} Logarithmic Potentials with External Fields. 
{\it Springer, 2010.}

\bibitem{Szego} {\sc G. Szeg\H{o}.} Orthogonal Polynomials. {\it American Mathematical Society, 1974.}

\bibitem{Vanlessen} {\sc M. Vanlessen.} Strong Asymptotics of Laguerre-Type Orthogonal Polynomials and Applications in Random Matrix Theory. {\it Const. Approx. 25, 2 (2007), 125--175.}

\bibitem{WF12} {\sc N. S. Witte and P. J. Forrester.}
On the variance of the index for the {G}aussian unitary ensemble. {\it Random Matrices: Theory Appl. \textbf{01}, 1250010 (2012).}\textbf{}

\bibitem{ZCD} {\sc Y. Zhao, L.H. Cao, D. Dai}. Asymptotics of the partition function of a
Laguerre-type random matrix model. {\it J. Approx. Theory 178 (2014), 64--90.}

\end{thebibliography}

\end{document}